\documentclass[12pt,a4paper,reqno]{amsart}
\usepackage[T1]{fontenc}
\usepackage{graphicx}
\usepackage{mathtools}
\usepackage{amssymb}
\usepackage{amsthm}
\usepackage{thmtools}
\usepackage[dvipsnames]{xcolor}
\usepackage{nameref}
\usepackage{hyperref}
\usepackage{natbib}
\usepackage{mathtools}
\usepackage[T1]{fontenc}
\usepackage[margin=1in]{geometry}
\usepackage{setspace}
\usepackage{soul}

\definecolor{darkblue}{rgb}{0, 0, 0.5}
\hypersetup{
     colorlinks = true,
     citecolor = Maroon,
     urlcolor = darkblue,
     linkcolor = darkblue
}

\newtheorem{definition}{Definition}
\newtheorem{lemma}{Lemma}

\newtheorem{proposition}{Proposition}

\theoremstyle{remark}
\newtheorem{example}{Example}
\newtheorem{remark}{Remark}
\newtheorem{condition}{Condition}
\newtheorem{assumption}{Assumption}[section]

\newtheoremstyle{example_contd}
{1em} {\topsep}%
{\upshape}% Body font
{}% Indent amount (empty = no indent, \parindent = para indent)
{\itshape}% Thm head font
{.}% Punctuation after thm head
{1em}% Space after thm head (\newline = linebreak)
{\thmname{#1} \thmnumber{ #2}\thmnote{#3} (continued)}% Thm head spec

\theoremstyle{example_contd}
\newtheorem*{example_contd}{Example}

\title[Relating inference methods for few treated units]{On the relationship between prediction intervals, tests of sharp nulls and inference on realized treatment effects in settings with few treated units}
\author{Luis Alvarez \and Bruno Ferman}
\begin{document}

\onehalfspacing

\begin{abstract}

We study how inference methods for settings with few treated units that rely on treatment effect homogeneity extend to alternative inferential targets when treatment effects are heterogeneous -- namely, tests of sharp null hypotheses, inference on realized treatment effects, and prediction intervals. We show that inference methods for these alternative targets are deeply interconnected: they are either equivalent or become equivalent under additional assumptions. Our results show that methods designed under treatment effect homogeneity can remain valid for these alternative targets when treatment effects are stochastic, offering new theoretical justifications and insights on their applicability.

\end{abstract}

	\maketitle
	\section{Introduction}

Inference on treatment effects in settings with few treated units presents significant challenges \citep{alvarez2025inferencetreatedunits}. When the number of treated units is small --- with the extreme case being a single treated unit --- there is limited information on potential outcomes under treatment, making uncertainty quantification a difficult task. As a consequence, inference methods that remain valid in the presence of few treated units often rely on strong assumptions, particularly regarding treatment effect heterogeneity. In this paper, we focus on methods derived in ``model-based'' settings in which treatment assignment is viewed as fixed (or conditioned on) and where uncertainty stems from stochasticity in potential outcomes.\footnote{See \cite{alvarez2025inferencetreatedunits} for a discussion on the appropriateness of such modeling in settings with few treated units} In this scenario, many of the existing approaches rely on treatment effect homogeneity assumptions, imposing that treatment effects are non-stochastic. Approaches that allow for stochastic treatment effects typically shift the inferential focus: instead of constructing confidence intervals for the average treatment effect on the treated (ATT), they focus on (i) testing sharp null hypotheses,  (ii) conducting inference on the realized treatment effect, or (iii) constructing prediction intervals. 

This first alternative, rather than testing hypotheses about the ATT, tests sharp null hypotheses, which assess whether the treatment had no effect on any treated unit.\footnote{More generally, these tests can also be used to test a null that the treatment effect takes a specific value for each treated unit, with probability one.} While commonly used in design-based approaches \citep{Imbens_matching,young_QJE}, this strategy has also been applied in model-based settings \citep[e.g.][]{Chung2021,Bugni2018,LeeShaikh2014,Heckman2024}. This second alternative seeks to provide inferential statements concerning the realization of treatment effects in the sample at hand, instead of averaging over possible realizations of treatment effects. This inferential target is particularly relevant when the goal is to  understand the specific context in which treatment was delivered rather than to generalize to other settings.  There are many papers that consider treatment effects as deterministic but potentially heterogeneous across individuals, which can be implicitly seen as a setting in which we condition on the stochastic treatment effects \citep[e.g.][]{conley20211inference,carvalho2018arco,ferman2019inference,synthetic_did,alvarez2023extensions,alvarez2023inference,chernozhukov2024ttest}. Finally, this third alternative constructs set-valued functions of the data --- prediction intervals --- that aim to cover a random variable with a given confidence over repeated samples. These intervals have a long tradition in the forecasting literature (e.g. \cite{Phillips1979,Brockwell1991}). More recently, this type of construction has been considered in causal inference settings with the explicit goal of constructing bands that cover the in-sample treatment effects on the treated (viewed as random variables) with a pre-specified probability \citep[e.g.,][]{Candes,kivaranovic2020conformal,chernozhukov2021exact,CWZ2021Pnas,cattaneo2021prediction,cattaneo2023uncertainty}. 

This note provides a unifying perspective on these different inferential approaches for settings with few treated units. As a leading example, we consider first a difference-in-means estimator and then extend our results to more general estimators. We show that all these inferential approaches are deeply interconnected: they are either equivalent under the same set of assumptions required for their individual validity or become equivalent under additional assumptions. Some of these equivalences are straightforward, while others require novel results that we derive in this paper.

First, we show that inference methods for the ATT that remain valid with few treated units  under homogeneous treatment effects are also valid for testing sharp null hypotheses, and that, conversely, any test of sharp nulls provides a valid test of the ATT under treatment effect homogeneity. This follows directly from the fact that imposing a sharp null is equivalent to assuming, under the null, that treatment effects are non-stochastic and homogeneous. 

Next, we show that a broad class of inference methods originally developed under the assumption of non-stochastic and homogeneous treatment effects can also be applied in settings with stochastic treatment effects for inference on the realized treatment effect, under strong assumptions --- a sufficient one being independence between the (stochastic) treatment effects and the potential outcomes under no treatment. To the best of our knowledge, this discussion of the assumptions required for valid inference on the realized treatment effect is novel.

We also derive new results providing conditions under which inference methods valid for the realized treatment effect, when treatment effects are independent of untreated potential outcomes, can also be used to construct prediction intervals, even when treatment effects and untreated potential outcomes exhibit arbitrary dependence.  This is the case when we have a consistent estimator for the  distribution of the untreated potential outcomes of the treated (without conditioning on the realized treatment effect). 

Finally, we show that prediction interval construction and sharp null hypothesis testing are intrinsically linked: under a broad class of inference methods used in small-sample settings, each can be obtained by ``inverting'' the other, remaining valid under the same set of assumptions. While such equivalence has been noted  in specific settings (e.g. \citet{Lei2013} and \citet{chernozhukov2021exact}), we believe that, by laying it out in a general setting, we contribute to clarifying the interpretation of prediction intervals. In particular, we show that the usual practice of checking whether prediction intervals contain zero can always be seen as a valid test of the sharp null that treatment effects are homogeneous and equal to zero in the treated population, when uncertainty stems from sampling from a well-defined population. The latter connection is not specific to a few treated setting, being true in any situation where prediction intervals for treatment effects are reported.

Overall, by establishing the connections among these approaches, this paper provides new theoretical justifications for inference methods originally developed under the assumption of homogeneous treatment effects (for example, \cite{conley20211inference} and \cite{ferman2019inference}), demonstrating their validity for alternative inferential targets in settings where treatment effects are stochastic.  Specifically, we show that, in settings with stochastic treatment effects, inference methods developed under treatment effect homogeneity (i) remain valid for testing sharp null hypotheses; (ii) remain valid for inference on the realized treatment effect under strong assumptions on the dependence  between treatment effects and untreated potential outcomes; and (iii) can be used to construct prediction intervals under arbitrary dependence between treatment effects and untreated potential outcomes, subject to additional assumptions.

\section{Leading example}
\label{leading_example}
The aim of this section is to illustrate our main results by means of a simple example. 

\subsection{Setup} Consider a setting where a researcher has access to a sample with $N$ units from a population of interest. There is an outcome of interest $Y$ and a policy intervention affecting a subset of the units in the population. We denote by $D$ the indicator for the units exposed to the intervention.

We define potential outcomes $Y(1)$ and $Y(0)$, corresponding to the outcomes that would have been observed for a unit in the population were the unit assigned, respectively, treatment and non-treatment. Observed outcomes are thus given by $Y = Y(0) + D(Y(1) - Y(0))$. We consider a model-based approach, in which we condition on treatment assignment and focus on uncertainty coming from potentially unobservable shocks that determine the potential outcomes
(see \cite{alvarez2025inferencetreatedunits} for further discussion on the use of model-based designs in settings with few treated units).

    Consider the difference-in-means estimator that is computed with a sample of $N$ units, with $N_1$ being treated and $N_0$ not. This estimator is given by:
\begin{eqnarray*}
    \hat\beta_{\text{DM}} &=& \frac{1}{N_1}\sum_{i=1}^N D_i Y_{i} - \frac{1}{N_0}\sum_{i=1}^N (1-D_i) Y_{i}  \\ &=& \frac{1}{N_1}\sum_{t=1}^T D_i \alpha_i + \frac{1}{N_1}\sum_{i=1}^N D_i Y_{i}(0) -  \frac{1}{N_0}\sum_{i=1}^N (1-D_i) Y_{i}(0),   
\end{eqnarray*}
where $\alpha_i = Y_{i}(1)-Y_{i}(0)$ is the individual treatment effect.

Under the mean independence assumption, $\mathbb{E}[Y_{i}(0)|D_1,\ldots, D_N] =\mu_0$, for every $i=1,\ldots, N$, $\hat{\beta}_{\text{DM}}$ is unbiased for the expected sample average treatment effect on the treated, i.e. 
$$\mathbb{E}[\hat\beta_{\text{DM}}|D_1,\ldots, D_N] = \frac{1}{N_1}\sum_{i=1}^N D_i \mathbb{E}[\alpha_i|D_1,D_2,\ldots D_N] \eqqcolon \beta. $$

The parameter $\beta$ captures the average expected effect of the intervention on the treated units in the sample, where expectations are taken with respect to the distribution of economic uncertainty determining potential outcomes.

\begin{example}
\label{example_leading}
    To make things more concrete, we illustrate this setting with an example in which the units are farmers and the treatment is agricultural technical assistance, with $Y$ denoting an agricultural outcome of interest. Uncertainty about $Y_i(0)$ and $\alpha_i$ arises from factors such as weather shocks, health shocks, and other random events. Some of these shocks may be independent across farmers, while others may be correlated. The mean independence assumption is well justified in this context because treatment was randomly assigned at the farmer level.

\end{example}

\begin{remark}
    Our setting nests difference-in-differences designs as a particular case, in which case we take $Y$ to be the difference of outcomes between two periods, and assume that treatment only affects a subset of the units in the second period, with $D$ being the indicator of being among these units.
\end{remark}

\subsection{Challenges for uncertainty quantification when $N_1$ is small} Even though $\hat \beta_{\text{DM}}$ is unbiased for $\beta$, it is difficult to quantify uncertainty regarding it when $N_1$ is small. To see this, observe that:

\begin{equation}
\label{eq_decomposition}
\hat \beta_{\text{DM}} - \beta = \frac{1}{N_1}\sum_{t=1}^T D_i (\alpha_i-\beta) + \frac{1}{N_1}\sum_{i=1}^N D_i Y_i(0) - \frac{1}{N_0}\sum_{i=1}^N (1-D_i)Y_i(0). 
\end{equation}

Uncertainty regarding $\hat \beta_{\text{DM}}$ can be decomposed into uncertainty concerning treatment effect heterogeneity on the treated units, $D_i(\alpha_i - \beta)$, uncertainty regarding untreated potential outcomes in the treatment group, $D_iY_{i}(0)$, and uncertainty regarding untreated potential outcomes in the control group, $(1-D_i)Y_{i}(0)$.

\subsection{Existing solutions for quantifying uncertainty of untreated potential outcomes}
\label{Sec: solutions for Y(0)}

In a setting with a fixed number of treated units but many controls, we can estimate the distribution of controls $(1-D_i)Y_{i}(0)$ when $N_0\rightarrow \infty$ under weak dependence assumptions \citep{conley20211inference,alvarez2023inference}. Moreover, if we assume that the distribution of $D_iY_{i}(0)$ for the treated is the same as the distribution of $(1-D_i)Y_{i}(0)$ for the controls, then we can \textit{extrapolate} the information from the controls to quantify uncertainty on the untreated potential outcomes of the treated, even when we have only a single treated observation. For example, we can strengthen  the mean independence assumption to the following assumption.

\begin{assumption}
    \label{ass_iid}
    $Y_{i}(0)|D_1,\dots, D_N \overset{\text{iid}}{\sim} F_{Y(0)|D_1,\ldots, D_N}$
\end{assumption}

In this case, the distribution of potential outcomes in the control group may be used to infer the distribution in the treatment group \citep{conley20211inference}. This assumption can be relaxed to accommodate  heteroskedasticity depending on a set of observed covariates, as shown in \cite{ferman2019inference}. See \cite{alvarez2025inferencetreatedunits} for other alternatives, including panel data settings in which such extrapolations may come from pre-treatment periods.

\subsection{Existing solutions for quantifying uncertainty of treatment effect heterogeneity}

Quantifying this type of uncertainty is more complicated, because it depends on potential outcomes under treatment, for which, in settings with few treated units, we observe limited information. In the extreme case in which there is only a single treated unit, we observe only a single observation of $Y_i(1)$.

We describe here common alternatives that have been considered in the literature to circumvent this problem in settings with few treated units (including the case with $N_1=1$).

\subsubsection{Assume treatment effect homogeneity}

\label{sec_homogeneous}
In this case, we assume that there exists a constant $\alpha$ such that  $\alpha_i=\alpha$ for every treated unit $i$. This means that the distribution of the $Y_i(1)$ for treated units can be obtained from a (common) location shift from the distribution of $Y_i(0)$. In this case, for each unit $i$, the treatment effect would be the same across all realizations of the uncertainty that determine the potential outcomes, and homogeneous across units.

Under this (arguably strong) assumption, $\beta = \alpha$, and variability due treatment effect heterogeneity disappears from the distribution of $\hat \beta_{\text{DM}}$. Therefore, if we account for the distribution of $Y_i(0)$ for both treated and control units using one of the methods discussed in Section \ref{Sec: solutions for Y(0)}, then we can conduct valid inference.

For example, consider the case in which only $i=1$ is treated. Under treatment effect homogeneity and  Assumption \ref{ass_iid}, \cite{conley20211inference} note that, when $N_0 \rightarrow \infty$, $(\hat \beta_{\text{DM}} - \beta)  \overset{d}{\rightarrow}  (Y_1(0) - \mathbb{E}[Y_1(0)])$, where the distribution of $(Y_1(0) - \mathbb{E}[Y_1(0)])$ can be estimated using the control residuals to construct p-values that are asymptotically valid when $N_1$ is fixed and $N_0 \rightarrow \infty$. More specifically, if we set $\hat p_{CT} = \frac{1}{N} \sum_{i=1}^N \mathbf{1}\left\{ |\hat \beta_{\text{DM}}-c| \leq |\hat u_i | \right\}$, where $\hat u_i$ are the residuals, this would provide an asymptotically valid p-value (when $N_0 \rightarrow \infty$) for the null $H_0: \beta = c$.

\begin{example_contd}[\ref{example_leading}]
    In this example, treatment effect homogeneity would mean that the treatment  would have the same causal effect on agricultural outcomes  irrespectively of whether we had a positive or negative weather shock, which may be an unreasonable assumption in many settings. We can imagine settings in which the treatment effect would be stronger when farmers were hit by a negative shock, which would invalidate the assumption of treatment effect homogeneity.  
\end{example_contd}

\subsubsection{Tests of sharp null}

Another possibility consists of considering a different null hypothesis. Instead of testing a hypothesis regarding the ATT, $\beta$, we could test a sharp null \citep[e.g.][]{Chung2021,Bugni2018}. In this case, the researcher is interested in testing, for some $c \in \mathbb{R}$:

\begin{equation}
\label{eq_sharp}
H_0: \mathbb{P}\left[\alpha_i = c|D_1,\ldots, D_N\right] = 1, \quad \forall i =1,\ldots N, \text{ s.t. } D_i=1\, .
\end{equation}

This null implies that treatment effects on treated units are constant over repeated realizations of sampling uncertainty, and equal to some value $c$.  More commonly, by setting $c=0$ we have a null that treatment effect has no effect whatsoever, meaning that for (almost) every realization of uncertainty on potential outcomes, the treatment has no effect on any treated unit. 

Under the null \eqref{eq_sharp}, Equation \eqref{eq_decomposition} subsumes to:

$$\hat \beta - c =   \frac{1}{N_1}\sum_{i=1}^N D_i Y_i(0) - \frac{1}{N_0}\sum_{i=1}^N (1-D_i)Y_i(0)  \, .$$

Consequently, we can construct valid tests for the sharp null by leveraging methods that quantify the uncertainty in untreated potential outcomes (Section \ref{Sec: solutions for Y(0)}). For example, under Assumption \ref{ass_iid}, the statistic 

\begin{equation}
    \label{eq_pvalue}
\hat{p}_c  = \hat{G}(|\hat \beta - c|)\, ,
\end{equation}
where
$$\hat{G}(x) = \frac{1}{|\Pi|}\sum_{\pi \in \Pi} \mathbf{1}\left\{\left| \frac{1}{N_1}\sum_{i=1}^N D_{i} (Y_{\pi(i)} - c D_{\pi(i)} )- \frac{1}{N_0}\sum_{i=1}^N (1-D_i)(Y_{\pi(i)} - c D_{\pi(i)} ) \right|\leq x\right\}\, ,$$
and $\Pi$ is the set of permutations on $\{1,\ldots, N\}$, is a valid p-value for testing the sharp null. This property follows from standard results on randomization tests (see Appendix A of \citet{alvarez2025inferencetreatedunits} for details). When Assumption \ref{ass_iid} does not hold, it is not generally possible to construct exact tests. However, when other restrictions enabling extrapolation from the control group are assumed it can be possible to construct tests that are asymptotically valid as $N_0\rightarrow \infty$ \citep[e.g.][]{ferman2019inference}.

\begin{example_contd}[\ref{example_leading}]
    In our example, by setting $c=0$, we would be testing a null hypothesis that the technical assistance has no effect whatsoever on agricultural outcomes of the treated units, regardless of the realizations of the random variables that determine the potential outcomes. 
\end{example_contd}

\begin{remark}
The sharp null \eqref{eq_sharp} is related, albeit distinct, from the sharp nulls considered in design-based settings \citep[e.g.][]{Imbens_Rubin_2015,young_QJE}. In those settings, potential outcomes in the sample are usually viewed as fixed (nonrandom) quantities, uncertainty stems from assignment, and the null being tested is:

$$H_0: Y_i(0) = Y_i(1) + c\, ,\quad i=1,\ldots, N\, .$$
This is a null that imposes that the treatment effects for all units in the sample are homogeneous and constant to a value $c$. In contrast, in model-based settings, the sharp null usually refers to constant and homogeneous effects for the treated units. 
\end{remark}

\subsubsection{Inference on the realized treatment effect} 

A third alternative for inference in settings with few treated units consists in targeting the \emph{realized sample average treatment effect}, i.e. one is concerned in constructing a test statistic for the null

\begin{equation}
\label{eq_realized}
     H_0 : \frac{1}{N_1}\sum_{t=1}^T D_i\alpha_i = c
\end{equation}
that controls size \emph{conditionally on the realized treatment effects} $\mathcal{A} = \{\alpha_i: D_i=1\}$, i.e. one considers a decision rule $\tilde{\phi}_c$ with the property that $\mathbb{E}[\tilde{\phi}_c|D_1,\ldots, D_n, \mathcal{A}] \leq \gamma$ if \eqref{eq_realized} holds. We can interpret settings that consider $\alpha_i$ as non-stochastic but potentially different across $i$ \citep[e.g.][]{conley20211inference,carvalho2018arco,ferman2019inference, synthetic_did,alvarez2023inference,alvarez2023extensions,chernozhukov2024ttest} as implicitly considering a setting with stochastic treatment effects but conditioning on the realized effects. In this case, constructing a test that controls size can be seen as equivalent, in a setting where effects are stochastic, to construct a test that controls size conditionally on $\mathcal{A}$. 

This inferential focus is especially pertinent when the objective is to learn about the treatment effect in the particular context where the intervention occurred, rather than to extrapolate to other environments. Pragmatically, in the extreme case of a single treated unit, ATT inference becomes infeasible if treatment effects are stochastic, since we only observe one realization of the treated potential outcome. In this case, focusing on uncertainty quantification for the realized treatment effects observed in the sample may be more feasible.

For inference on realized effects to be valid, one requires that the assumptions discussed in Section \ref{Sec: solutions for Y(0)} are valid even when we condition on $\mathcal{A}$. As we discuss in more detail in Section \ref{sec_homogeneity_realized}, in many settings this would require relatively strong assumptions.

\begin{example_contd}[\ref{example_leading}]
     In our example, we may consider settings in which the technical assistance is only helpful when there is a negative weather shock, so the effect could be $\tau$ if there is a negative weather shock in farm $i$ (which happens with probability $\pi$) and zero otherwise. In this case, the ATT would be $\tau \times \pi$, but we would be drawing inference on either $\tau$ or $0$, depending on whether or not we had a negative weather shock. In this setting, inference on the realized treatment effect could be interesting if the goal is understanding the effects of that particular implementation of the program. In contrast, if the goal is generalizing to the expected effect across different scenarios, then the ATT would be a more interesting target parameter.   In settings with only a single treated unit, note that we only observe a realization of the treated outcome either when there was or when there was not a negative weather shock. Therefore, there would not be much hope in drawing inferences on the ATT, but we could have some hope in drawing inferences on the realized treatment effect. 
\end{example_contd}

\begin{example_contd}[\ref{example_leading}]
    As another illustration, we can consider settings in which the quality of the technical assistance is stochastic, and that generates variability of the causal effect of such treatment. In this case, we would be considering inference on the treatment effect conditional on the quality of the technical assistance. 
\end{example_contd}

\subsubsection{Constructing prediction sets}

\label{sec_prediction_sets}

Another alternative is to construct prediction sets, rather than confidence intervals. Prediction sets are set-valued statistics ${C}$ with the property that, with a given confidence $1- \gamma \in (0,1)$:

\begin{equation}
\label{eq_set_property}
   \mathbb{P}\left[\frac{1}{N_1}\sum_{i=1}^N D_i \alpha_i\in \mathcal{C}\Big| D_1,\ldots, D_N\right]\geq 1-\gamma    \, .
\end{equation}

Prediction sets are popular in the conformal prediction literature \citep{Lei2013}, and have been considered in causal settings with few treated units  by \cite{cattaneo2021prediction,cattaneo2023uncertainty} and \cite{chernozhukov2021exact}. It has been argued that such construction is especially useful in settings with substantial treatment effect heterogeneity -- and where the goal is to personalize treatment allocation --, in which case assessing uncertainty regarding individual treatment effects may be more relevant than focusing on average treatment effects from the sampled population \citep{Candes,kivaranovic2020conformal}.  

A set-valued function of the data satisfying \eqref{eq_set_property} has the property of containing the sample average treatment effect on the treated with probability at least $1-\gamma$, over repeated realizations of sampling uncertainty. If one had knowledge of the quantile function $Q_{\bar{Y}(0)|D_1,\ldots,D_N}$ of the conditional distribution of $\frac{1}{N_1}\sum_{i=1}^{N_1} D_i Y_i(0)$ given $D_1,\ldots, D_N$, then it would be possible to  take $\mathcal{C}$ as the interval:

$$\left[\frac{1}{N_1}\sum_{i=1}^N D_i Y_i -Q_{\bar{Y}(0)|D_1,\ldots,D_N}(1-\gamma/2), \frac{1}{N_1}\sum_{i=1}^N D_i Y_i -Q_{\bar{Y}(0)|D_1,\ldots,D_N}(\gamma/2)\right]\, .$$

Since the quantile function $Q_{\bar{Y}(0)|D_1,\ldots,D_N}$ is typically unknown, one approach to constructing prediction intervals is to focus on estimators of these quantities with good statistical properties \citep{cattaneo2021prediction}. We further discuss this type of construction in Section \ref{sec_general}.

\subsection{Connections between different approaches}In this paper, we establish several connections between the four alternative inference approaches described earlier. In this section, we present these connections for the simpler case of a difference-in-means estimator. Some of the connections rely on new results that we derive in a more general setting in Section \ref{sec_general}. Others are not new, but we believe their interpretation has been somewhat underappreciated.

\subsubsection{Connection between tests assuming treatment effect homogeneity and tests of sharp nulls}

It is easy to see that the  sharp null \eqref{eq_sharp} is equivalent to the conditions $\beta=c$ and $\mathbb{P}[\cap_{i: D_i=1}\{\alpha_i = \beta\}|D_1,\ldots, D_N]=1$, where the first term states that the ATT is equal to $c$, while the second term means that treatment effects are non-stochastic and homogeneous. Consequently, under the sharp null \eqref{eq_sharp}, we have that the assumptions considered in the methods discussed in  Section \ref{sec_homogeneous} for inference on the ATT are valid. 

In other words, if we have a test that is valid for inference on the ATT under the assumption that there is no treatment effect heterogeneity, then such test would also be valid for testing a sharp null. Conversely, if we have a test that is valid for inference on a sharp null, then this test would also be valid for inference on the ATT under the assumption of no treatment effect heterogeneity.

\subsubsection{Connection between tests assuming treatment effect homogeneity and inference on the realized treatment effects}\label{sec_homogeneity_realized}
Consider a test of the null $H_0 : \beta = c$ that is valid under the assumption that there is no treatment effect heterogeneity at the $\gamma$ significance level. Let the test statistic be $\hat \phi = g\left(\left\{Y_i,D_i\right\}_{i=1}^N\right)$, where $g$ is a map that is not a function of the data.\footnote{The map $g$ may depend on other random variables that are independent from the data, such as in the case of randomized decision rules.} Assume that this statistic satisfies the following condition:

\begin{condition}
\label{cond_struct_1}
    The test statistic $\hat \phi$ can be represented as $\hat{\phi} = h\left(\{Y_i(0),D_i\}_{i=1}^N, \frac{1}{N_1}\sum_{i=1}^{N} D_i \alpha_i\right)$, where $h$ is a map that is not a function of the data.
\end{condition}

 In other words, Condition \ref{cond_struct_1} restricts consideration to tests whose conclusions depend on the actual treatment effect heterogeneity only through the average of effects on the treated $\frac{1}{N_1}\sum_{i=1}^{N} D_i \alpha_i$. 
 
 We show that a test satisfying Condition \ref{cond_struct_1} that is valid under treatment effect homogeneity at the $\gamma$ significance level   would also be valid for inference on the realized treatment effect under assumptions that limit dependence between treatment effects and untreated potential outcomes. Specifically, under the assumption:

\begin{assumption}
    \label{ass_ind}  
$\mathcal{A} \text{ is independent of } \{Y_i(0)\}_{i=1}^n |D_1,\ldots, D_N\, ,$
\end{assumption}
\noindent the test is valid for testing the null $H_0: \frac{1}{N_1}\sum_{i=1}^{N_1}D_i \alpha_i=c$ conditionally on $\mathcal{A}$, since, on the event that the null holds, 

\begin{equation*}
    \begin{aligned}
        \mathbb{E}[\hat{\phi}|D_1,\ldots, D_N,\mathcal{A}] = \mathbb{E}[h(\{Y_i(0),D_i\}_{i=1}^N, c)|D_1,\ldots, D_N,\mathcal{A}] =\\
        \mathbb{E}[h(\{Y_i(0),D_i\}_{i=1}^N, c)|D_1,\ldots, D_N]\leq \gamma \, .
    \end{aligned}
\end{equation*}

The intuition for the above result is that, conditional on $\mathcal{A}$, we have a model with homogeneous treatment effects where the distribution of $\{Y_i(0)\}_{i=1}^n |D_1,\ldots, D_N$ satisfies the properties required for valid inference when there is no treatment effect heterogeneity.

Notice that Assumption  \ref{ass_ind} is potentially restrictive: in particular this assumption is not necessarily valid  even if we assume that $\{Y_i(0)\}_{i=1}^N$ is iid and  treatment is randomly assigned. In contrast, Assumption \ref{ass_iid} would be satisfied under these conditions.

Assumption \ref{ass_ind} is satisfied if we consider a model in which (stochastic) treatment effects $\mathcal{A} = \Gamma(U)$ are a function of unobserved variables $U$, while potential outcomes when untreated $\{Y_i(0)\}_{i=1}^n = \Lambda(V)$ is a function of unobserved variable $V$. In this case, Assumption \ref{ass_ind} means that the sources of variability that determine $\mathcal{A}$ are independent from the sources of variability that determine $\{Y_i(0)\}_{i=1}^n = \Lambda(V)$, that is, $U \perp V | D_1,...,D_N$.

Another alternative is that there is dependence between $\mathcal{A}$ and $\{Y_i(0)\}_{i=1}^n$ (conditionally on $D_1,\ldots, D_N$), but we have that the distribution $\{Y_i(0)\}_{i=1}^n |\mathcal{A},D_1,\ldots, D_N$ satisfy the assumptions necessary so that we can learn about $Y_i(0)$ of the treated using the outcomes from the controls, e.g. if  we replace Assumption \ref{ass_iid} with $Y_i(0)  \overset{iid}{\sim} F_{Y(0)|\mathcal{A},D_1,\ldots,D_N}$.

\begin{example_contd}[\ref{example_leading}]
    Considering again  our leading example, we can think that the independence between treatment effects and potential outcomes when untreated is reasonable when the heterogeneity in treatment effects come from variability in the quality of the institution implementing the treatment. However, this assumption would be less reasonable if the heterogeneity in treatment effects comes from variables that may be related to potential outcomes when untreated, such as weather shocks. The key point in this case is whether we have different weather shocks for each farm, or we have a single weather shock that affects all farms in  a similar way. In the first case, even if Assumption \ref{ass_iid} is valid unconditionally, it would be hard to justify that it would also be valid conditional on $\mathcal{A}$. In the second case, it should be more reasonable to assume that Assumption \ref{ass_iid} is reasonable conditional on $\mathcal{A}$, even if we have that the treatment effect heterogeneity is potentially related to the potential outcomes when untreated.
\end{example_contd}

\begin{example_contd}[\ref{example_leading}]
    Consider again the case treatment effect depends only observable whether shocks, which we denote by  $Z_i$. In this case, it might be reasonable to assume that Assumption \ref{ass_iid} is valid conditional on the treatment effect, if we also condition on $Z_i$. In this case, one would conduct inference on the realized treatment effect separately for each subsample defined by the values of $Z_i$.
\end{example_contd}

\begin{remark}
\label{rmk_studentize}
    Considering this difference-in-means setting, a test based on the statistic $\hat \beta_{DM}-c$ would satisfy Condition \ref{cond_struct_1}. This is the case of the procedures in  \cite{conley20211inference} and \cite{ferman2019inference}. In contrast, a test based on an  unequal variance t-statistic would not satisfy this condition, since the conclusion of the test would also depend on the dispersion of the $\alpha_i$. In this case, even though the test is valid for inference on the ATT $\beta$ under treatment effect homogeneity, it would not be generally valid, in settings with $N_1$ fixed, for inference on the realized treatment effect under heterogeneity, even if the independence Assumption \ref{ass_ind} holds. This fact has been noted elsewhere in the few treated literature \citep{Chaisemartin2022,dechaisemartin2024book,alvarez2025inferencetreatedunits}, even though it is recognized in these settings that, in standard asymptotic analyses where $N_1,N_0\to\infty$, studentized statistics are preferable due to their robustness to heterogeneity when the goal is to infer the ATT $\beta$.
\end{remark}

\begin{remark}
\label{rmk_det_hete}
Consider a setting where, in the language of \cite{alvarez2025inferencetreatedunits}, we have deterministic heterogeneous treatment effects, meaning that the $\alpha_i$ are fixed constants (over repeated samples), although possibly distinct across $i$. In this case, Assumption \ref{ass_ind} is trivially satisfied. As a consequence, the connections we establish in this section immediately extend to the case where treatment effect heterogeneity is deterministic, by casting inference on the ATT under deterministic heterogeneity as a special case of inference on the realized effect (where Assumption \ref{ass_ind} holds by construction). In particular, any test statistic satisfying Condition \ref{cond_struct_1} that produces a valid test for the ATT under treatment effect homogeneity will also be valid for inference on the ATT under deterministic treatment effect heterogeneity.
\end{remark}

\subsubsection{Connection between tests of sharp nulls and prediction sets} Tests of sharp nulls and prediction sets are intimately connected. As an example, for a given significance level $\gamma \in (0,1)$, it is possible to invert the  procedure given by the decision rule $\phi_c = \mathbf{1}\{\hat{p}_c \leq \gamma\}$, where $\hat{p}_c$ denotes the p-value in \eqref{eq_pvalue}, to construct a set $\mathcal{I} = \{c \in \mathbb{R}: \phi_c = 0\}$. Doing so will result in a {prediction set} that satisfies \eqref{eq_prediction_few_treated} when Assumption \ref{ass_iid} is satisfied.

More generally, in the next section, we show that,  any set-valued function satisfying \eqref{eq_set_property} defines a valid test of sharp null \eqref{eq_sharp}; and, conversely, for the types of test statistics typically considered in settings with few treated units, inverting a test of the sharp null \eqref{eq_sharp} will result in a prediction set satisfying \eqref{eq_set_property}. While unsurprising, this result seems to be unappreciated in the literature. Specifically, it provides a rationale for reporting prediction sets of treatment effects in settings with few treated units, as doing so is equivalent to reporting those values for which a test of a sharp null is not rejected by the data. Finally, we note that the class of tests whose inversion leads to valid prediction intervals is quite large. Indeed,  consider a valid family of tests of sharp nulls $\{\hat{\phi}_c\}_{c \in \mathbb{R}}$ that satisfy the following condition:

\begin{condition}
    \label{cond_struct_2}
    The family of tests of sharp nulls $\{\hat{\phi}_c\}_{c \in \mathbb{R}}$ may be represented as $$\hat{\phi}_c = v\left(\{Y_i(0),D_i\}_{i=1}^N, \frac{1}{N_1}\sum_{i=1}^{N_1}D_i \alpha_i-c\right),\quad \forall c \in \mathbb{R}\, ,$$ where $v$ is not a function of the data.
\end{condition}

In this case, our results show that inversion of this family of tests will result in valid prediction sets. Condition  \ref{cond_struct_2} requires that the conclusion of the tests $\hat{\phi}_c$ depend on the in-sample average treatment on the treated $\frac{1}{N_1}\sum_{i=1}^{N_1}D_i \alpha_i$ and the null value $c$ \textit{only through} the difference $\frac{1}{N_1}\sum_{i=1}^{N_1}D_i \alpha_i-c$. Notice that, if we have a family of tests that satisfy Condition \ref{cond_struct_2}, then each individual test has the structure in Condition \ref{cond_struct_1}.

\subsubsection{Connection between  prediction intervals and inference on realized treatment effects}  \label{sec_pred_realized}

If we have a test that is valid for inference on the realized treatment effects, then it is clear that it would also be valid to construct prediction intervals. The intuition is that a confidence interval for the realized treatment effect would have valid coverage conditional on the treatment effect, so when we integrate over the distribution of treatment effects it would lead to a valid prediction interval. 

The other direction, however, is more tricky. It might be that we have a valid prediction interval, but it does not lead to valid confidence intervals for the realized treatment effects (that is, once we condition on the treatment effects). The reason is that we may have some realizations of the treatment effects in which the prediction interval has conditional (on treatment effects) undercoverage, which is compensated by other realizations of the treatment effects in which we have over-coverage. In Appendix \ref{example_failure} we provide a simple example in which this may happen (see also \cite{CWZ2023_comment} for a related discussion). 

Therefore, we do not have a direct equivalence between prediction intervals and inference on the realized treatment effects. Still, we show that the following result is valid:  for the types of test statistics typically considered in the few treated literature,\footnote{Specifically, we consider inference methods based on statistics of the form $\frac{1}{N_1}\sum_{i=1}^{N_1} D_i (Y_i-\hat M_i)$, where $\hat{M}_i$ is a consistent estimator of a proxy $\hat M_i$ of the untreated potential outcome $Y_i(0)$. These cover several methods in the few treated literature. See Section \ref{eq_specialized} for further details.} if we have a test that is valid for the realized treatment effect under the assumption that the treatment effects are independent of untreated potential outcomes (Assumption \ref{ass_iid}), then inversion of this test will also produce a valid prediction interval, even when we allow for arbitrary dependence between the stochastic treatment effects and the potential outcomes when untreated.

Put another way, our results provide conditions ensuring that, if one has a consistent estimator of the distribution of $\frac{1}{N_1}\sum_{i=1}^N D_i Y_i(0)$ (conditionally on assignments), then it is possible to construct an interval that is simultaneously an (asymptotically) valid prediction interval {with no restriction on the dependence between treatment effects and untreated potential outcomes}, as well as an (asymptotically) valid confidence interval for the realized treatment effects under the additional assumption of independence between treatment effects and untreated potential outcomes. For example, we would be able to construct a consistent estimator of the distribution of $\frac{1}{N_1}\sum_{i=1}^N D_i Y_i(0)$ (conditionally on assignments) if $Y_i(0)$ has the same marginal distribution across $i=1,...,N$ and $\{Y_i(0)\}_{i:D_=0}$ is weakly dependent \citep{alvarez2023inference,conley20211inference}. 

\begin{example_contd}[\ref{example_leading}]
    In our empirical illustration, imagine that we have weather shocks that may affect agricultural outputs. If we are in a setting in which all farms are in the same geographical region, then we should expect  weather shocks that affect all units {in a similar way}, so the assumption for the connection between prediction intervals and inference on realized treatment effects would not be valid,  because in this case it would be difficult to find a consistent estimator for the conditional-on-assignment distribution of $\frac{1}{N_1}\sum_{i=1}^N D_i Y_i(0)$. In contrast, this assumption would be more reasonable if the farms are geographically more distant, so they may experience different weather shocks.   
\end{example_contd}

\begin{remark}
    Importantly, while a consistent estimator for the conditional-on-assignment distribution of $\frac{1}{N_1}\sum_{i=1}^N D_i Y_i(0)$ guarantees valid prediction intervals without further assumptions and valid inference on the realized treatment effects under Assumption \ref{ass_ind}, there may be settings where valid inference on the realized treatment effect -- and consequently valid prediction intervals -- may be plausible even though a consistent estimator of the  conditional-on-assignment distribution of  $\frac{1}{N_1}\sum_{i=1}^N D_i Y_i(0)$ is unavailable. For example, consider a setting in which we have dependence between $\{Y_i(0)\}_{i=1}^N$ and $\mathcal{A}$, but we have that Assumption \ref{ass_iid} is valid conditional on $\mathcal{A}$. In this case, inference on the realized treatment effect would be valid (and, therefore, we would also have valid prediction intervals), while we do not necessarily have a consistent estimator for the distribution of $\frac{1}{N_1}\sum_{i=1}^N D_i Y_i(0)$.
\end{remark}

\begin{remark}
Relating to the discussion of deterministic heterogeneity in Remark \ref{rmk_det_hete},  the results in this section imply that, for the types of constructions typically employed in the few treated literature, confidence intervals for $\beta$ derived under deterministic treatment effect heterogeneity can be reinterpreted as valid prediction intervals for $\frac{1}{N_1}\sum_{i=1}^{N_1} D_i \alpha_i$, in a setting where treatment effects exhibit stochastic heterogeneity. This type of connection has also been noted by \citet{chernozhukov2025debiasingttestssyntheticcontrol} in a setting where proxies for missing potential outcomes are constructed through the synthetic control method and the goal is to conduct inference on average effects.
\end{remark}

\subsubsection{Connection between  sharp nulls and inference on realized treatment effects}  

Since we show that sharp nulls and prediction intervals are intimately connected, this means that we also have a connection between tests of sharp nulls and inference on the realized treatment effects. {Indeed, note that inference on the realized treatment effects under Assumptions \ref{ass_iid} and \ref{ass_ind} will result in \emph{exactly the same} p-value $\hat{p}_c$ given in \eqref{eq_pvalue} as the test of sharp nulls under Assumption \ref{ass_iid}. Therefore, the decision rule to reject the null if $\hat{p}_c$ is below some significance value produces a test of the sharp null \eqref{eq_sharp} under Assumption \ref{ass_iid}, and a test on realized average affects \eqref{eq_realized} under Assumptions \ref{ass_iid} and \ref{ass_ind}. As we show in the next section, this is a more general feature of tests designed for settings with few treated units: under some assumptions, a test that is asymptotically valid for testing the sharp null \eqref{eq_sharp}, is also valid for testing \eqref{eq_realized} under the Assumption \ref{ass_ind} that restricts the relation between treatment effects and untreated potential outcomes.}

\section{General results}
\label{sec_general}
    
    \subsection{Setting} 
    \label{general_setting}
    Fix a probability space $(\Omega, \Sigma, P)$. We consider a setting with $N_1$ units of interest, indexed by $i=1,\ldots N_1$. For each of these units, we define potential outcomes $Y_i(0)$ and $Y_i(1)$ corresponding to the thought experiment of manipulating a binary treatment over an outcome $Y$ of interest. In line with the model-based setting of earlier sections, we view $(Y_i(0),Y_i(1))$ as random variables defined on $(\Omega, \Sigma, P)$. Uncertainty on potential outcomes may reflect a sampling process from a population or, in those settings where such interpretation is not warranted, economic uncertainty that determines potential outcomes.
   
   {We remain agnostic about the interpretation of the unit index $i$: it may reflect different individuals (or counties, states etc.), or the same individual (or counties, states etc.) at different points in time, or reflecting both different individuals over different time periods}. This allows us to establish connection between inference approaches in both settings where the target are averages of effects across individuals and/or across the time. The causal effects of the policy on each unit are denoted by the random variables $\alpha_i=Y_i(1)-Y_i(0)$, $i=1$.

   We assume that the econometrician observes a sample $\{Y_i,X_i\}_{i=1}^{N_1} \cup \{\xi\}$ of random variables defined on $(\Omega, \Sigma, P)$, where $Y_i = Y_i(1)$, $i=1,\ldots N_1$. In other words, the econometrician observes the outcomes of the $N_1$ units under the treatment regime. The random variables $\{X_i\}_{i=1}^{N_1}$ denote a set of unit-level covariates used in the analysis, while $\xi$ denotes auxiliary data used in the imputation of the missing potential outcomes $Y_i(0)$ that is always required, either implicitly or explicitly, by any causal inference approach \citep[see][for a discussion]{athey2021matrix}. The auxiliary data $\xi$ can include past values of  the outcome $Y$ for the $N_1$ units -- when the assumptions underlying the causal inference approach enable extrapolation of the missing potential outcome from the past --, the values of $Y$ for a group of controls -- as in methods exploiting cross-sectional variation in assignment -- or a combination of both -- as in panel data methods such as differences-in-differences or the synthetic control method.

   In what follows, we denote the push-through probability measure of $(Y_i(0),Y_i(1),X_i)$ by $\mathbb{P}_i$. Expectations with respect to a probability law $Q$ is denoted by $E_Q$.

	\subsection{Three types of inferential targets}

 We now formalize, in the setup described in Section \ref{general_setting}, the alternative inferential targets introduced in Section \ref{leading_example}. We begin with a definition of a test of a sharp null.  In this section, we focus on definitions that assume finite-sample validity of the different inferential approaches, though it is immediate to extend these to asymptotic sequences where validity is only approximate. Asymptotic definitions are discussed in  Section \ref{eq_specialized}.

	\begin{definition}[Test of sharp null]
		Fix $c \in \mathbb{R}$. A test of the sharp null:
		
		$$H_0: \mathbb{P}_i[\alpha=c]=1\, ,\quad \forall i=1,\ldots,N_1\, ,$$
		at significance level $\gamma$ is a function of the sample $\phi_c$ taking values on $[0,1]$ such that:
		
		$$E_P[\phi_c]\leq \gamma \, ,$$
		whenever $H_0$ is true.
	\end{definition}
	
In our setting, a sharp null is a hypothesis on \textbf{both} the level of the treatment effects on the treated units, as well as on the absence of treatment effect heterogeneity across these units.
Notice that, when uncertainty stems from a sampling process from a well defined population, the sharp null implies that treatment effects are constant and homogeneous in the (treated) population from which the sample is drawn.
    
	\begin{definition}
		A $(1-\gamma)$-prediction set for the sample average treatment effect on the treated is a function of the sample $\mathcal{C}$, taking values on the subsets of $\mathbb{R}$, such that:
		
		$$P\left[\frac{1}{N_1}\sum_{i=1}^{N_1} \alpha_i \in \mathcal{C}\right]\geq 1-\gamma$$
	\end{definition}
	A prediction set is a function of the sample that, over repeated samples, contains the (possibly stochastic) sample average treatment effect in at least $100(1-\gamma) \%$ of the cases.

		The third type of inference concerns inference on \textbf{the realized sample average effect on the treated}.
        
			\begin{definition}
			A $(1-\gamma)$-confidence set for the realized sample average treatment effect on the treated is a function of the sample $\mathcal{I}$, taking values on the subsets of $\mathbb{R}$, such that:
			
			$$P\left[\frac{1}{N_1}\sum_{i=1}^{N_1} \alpha_i \in \mathcal{I}\Big|\alpha_{1}, \alpha_{2}, \ldots, \alpha_{N_1}\right]\geq 1-\gamma \, , \quad P-a.s.$$
		\end{definition}
		Clearly, by iterated expectations,  any valid confidence set for the realized effect is a valid prediction set for the sample average effect. We note, however, that, given its conditonal nature, construction of confidence sets for realized effects will typically require assumptions on the relation between treatment effects $\alpha$ and untreated outcomes $Y(0)$ in the treated population that may be hard to justify in practice. In the next sections, we show that inference procedures devised for inference on the realized effect may be reinterpreted as confidence sets for the sample average effect under a different set of assumptions that may be more palatable to applied researchers.
		\subsection{Relating tests of sharp nulls and prediction sets}
In this section, we lay out results that connect tests of sharp nulls and prediction sets. These results hold more generally, and are not restricted to few-treated-asymptotics. We discuss connections under a few-treated asymptotic sequence in the next section.

        The first result shows that prediction intervals always yield tests of sharp nulls.
	
	\begin{lemma}
		\label{lemma_imply}
		Let $\mathcal{C}$ be a valid $(1-\gamma)$ prediction set. We then have that, for every $c \in \mathbb{R}$, the decision rule $\phi_c = \mathbf{1}\{c \notin \mathcal{C}\}$ yields a valid test for the sharp null $\mathbb{P}_i[\alpha_i = c]=1, \ \forall i=1\ldots, N_1$, at significance level $\gamma$.
		\begin{proof}
			Under the null $P[\cap_{i=1}^{N_1}\{\alpha_i = c\}] = 1$:
			
			$$\mathbb{E}_P[\mathbf{1}\{c \notin \mathcal{C}\}] = P\left[\frac{1}{N_1}\sum_{i=1}^{N_1} \alpha_i \notin \mathcal{C}\right] \leq \gamma \, .$$
		\end{proof}
		\end{lemma}  
		
		The previous result is straightforward, though its interpretation may be novel. The result shows that the procedure of checking whether a value is within a prediction set for the sample average treatment effect may be seen as a test of a sharp null. 
        
        In what follows, we provide a partial converse to the above result. For that, we consider an analog of  Condition \ref{cond_struct_2} for this more general setting.

        \begin{condition}
        \label{cond_struct_general}
             Consider a family of decision rules $\{\phi_c: c\in \mathbb{R}\}$ with the property that, for some function $h$ that does not depend on the data, the tests can be represented as, for every $c \in \mathbb{R}$:
              $$\phi_c = h\left(\frac{1}{N_1}\sum_{i=1}^{N_1} \alpha_i - c, \{Y_i(0),X_i\}_{i=1}^{N_1} \cup \{\xi\}\right).$$
        \end{condition}

        For this class of decision rules, we show that the inversion of a testing procedure of sharp nulls generates a valid prediction interval.
        
        \begin{lemma}
        \label{lemma_partial_converse}
             Consider a family of decision rules $\{\phi_c: c\in \mathbb{R}\}$ that satisfy Condition \ref{cond_struct_general}.
          If $\mathbb{E}_P[h\left(0, \{Y_i(0),X_i\}_{i=1}^{N_1} \cup \{\xi\}\right)] \leq \gamma$, we then have that:

              \begin{enumerate}
                  \item $\{\phi_c: c\in \mathbb{R}\}$ is a family of valid tests for the sharp nulls $P[\cap_{i=1}^{N_1} \{\alpha_i=c\}]=1$, $c \in \mathbb{R}$, at the $\gamma$ significance level.
                  \item The region $\mathcal{C}=\{c \in \mathbb{R}: \phi_c = 0\}$ is a valid $(1-\gamma)$ prediction region for $\frac{1}{N_1}\sum_{i=1}^{N_1} \alpha_i$.
              \end{enumerate}
              \begin{proof}
                (1) Fix $c \in \mathbb{R}$. Under the null $P[\cap_{i=1}^{N_1} \{\alpha_i=c\}]=1$, $\phi_c = h\left(0, \{Y_i(0),X_i\}_{i=1}^{N_1} \cup \{\xi\}\right)$ a.s. Consequently, under such null, $E_P[\phi_c] = \mathbb{E}_[\hat{\psi}(0)]\leq \gamma$. (2) $P\left[\frac{1}{N_1}\sum_{i=1}^{N_1} \alpha_i \in  \mathcal{C}\right] = P\left[\phi_{\frac{1}{N_1}\sum_{i=1}^{N_1} \alpha^1_i } = 0\right]\geq 1-\gamma$.

              \end{proof}
        \end{lemma}

        The previous lemma shows that, for a rather general class of tests of sharp nulls, the inversion of these tests produces a valid prediction region for the sample average treatment effect on the treated $\frac{1}{N_1}\sum_{i=1}^{N_1} \alpha_i$. The assumption, spelled as Condition \ref{cond_struct_general}, that the decision rule depends on $\frac{1}{N_1}\sum_{i=1}^{N_1} \alpha_i$ and the hypothesized value $c$ only through the difference  $\frac{1}{N_1}\sum_{i=1}^{N_1} \alpha_i -c$ cannot be generally dispensed with. Indeed, in Appendix \ref{app_counter_equivalence}, we provide a simple example of a valid family of tests of sharp nulls that does not possess this structure and whose inversion does not produce a valid prediction interval for the average effect. Following the same discussion as in Remark \ref{rmk_studentize}, we also note that Condition \ref{cond_struct_general} would not be satisfied if we consider an unequal variance t-test.

        In the next section, we show that a broad class of tests employed in the few treated literature enjoys the structure in the statement of Lemma \ref{lemma_partial_converse}. As a consequence, these tests can be used both as (asymptotically) valid tests of sharp nulls, and may as well be inverted to construct (asymptotically) valid prediction sets for the sample average effect on the treated.

		\subsection{Conducting inference in few treated settings}
		\label{eq_specialized}
		\subsubsection{Environment} In this section, we specialize the structure of our problem to the one typically adopted in few treated settings. Specifically, we follow \cite{alvarez2025inferencetreatedunits} and decompose the unobserved potential outcome of treated observations as:
		
		$$Y_{i}(0) = M_i + \epsilon_i, \quad i=1,\ldots N_1 \, ,$$
		where $M_i$ is a proxy for the untreated potential outcome for the treated observations. Let $\hat{M}_i$ be an estimator for this proxy. The average effect may then be estimated as:
		
		$$\widehat{\boldsymbol{\alpha}} = \frac{1}{N_1} \sum_{i=1}^{N_1} (Y_i - \hat{M}_i)\, .$$
		
		In what follows, we shall assume that the estimators $\hat{M}_i$ adopted by the researcher are consistent in an asymptotic sequence where the number of treated units is fixed. Importantly, we shall remain agnostic on whether the estimator relies on a large time series or on a large number of controls (or both). This allows us to cover a wide range of methods.\footnote{See \cite{chernozhukov2021exact} and \cite{alvarez2025inferencetreatedunits} for examples of proxies that rely on time series or cross-sectional variation or both.}
        
        To formally state asymptotic inference results, we embed our setting in an asymptotic framework indexed by $s \in \mathbb{N}$, and consider the behavior of the inference procedures as $s \to \infty$. We allow both the sample law, now indexed by $s$ and denoted by $P_s$, as well as the dimension of the vector of auxiliary random variables, denoted by $\xi_s$, to vary with $s$. Doing so allows us to capture settings where consistent estimation of the proxies relies on a large number of pre-treatment variables or on a large number of controls (or both).\footnote{Moreover, by allowing the sample law to vary with $s$, it is possible, if one considers arbitrary sequences of laws $P_s \in \mathcal{P}_s$ in a sequence of spaces $\mathcal{P}_s$, $s \in \mathbb{N}$, to obtain uniform-in-law asymptotic coverage results \citep{Canay2017,Andrews2020}.}  In keeping with the few treated literature, we consider an asymptotic sequence where $N_1$ is fixed. 

        The assumption of consistent estimation of the proxies can be stated in our asymptotic framework as follows:
		
		\begin{assumption}
  \label{ass_cons_proxy}
			For each $i=1,\ldots, N_1$ and every $\epsilon > 0$:
			
			$$\lim_{s \to \infty} P_s[|\hat{M}_{i} -{M}_{i}|> \epsilon] = 0\, .$$
			\end{assumption} 
        
    \subsubsection{Constructing asymptotically valid tests of sharp nulls, prediction intervals and confidence intervals for the realized effect}  Notice that, under an asymptotic sequence that satisfies Assumption \ref{ass_cons_proxy}, we have that:
		
		 $$\widehat{\boldsymbol{\alpha}}  - \frac{1}{N_1}\sum_{i=1}^{N_1}\alpha_i = \frac{1}{N_1}\sum_{i=1}^{N_1}\epsilon_i + o_{P_s}{(1)} \, .$$
		
        Consequently, for asymptotically valid tests of sharp nulls, it suffices to approximate the distribution of the $\epsilon_i$.\footnote{Following \cite{cattaneo2021prediction}, we can also consider improvements to asymptotically valid inference methods that explicitly take into account the $o_{P_s}(1)$ estimation error of the proxy. Given that these constructions offer asymptotically vanishing improvements, they do not alter the main connections established in this section.} For example, let $u \mapsto \hat{Q}_s(u)$ be a pointwise consistent estimator of the quantile function $Q_s$ of the distribution of $\frac{1}{N_1}\sum_{i=1}^{N_1}\epsilon_i$, meaning that, for each $u \in (0,1)$:
        $$|\hat{Q}_s(u)-Q_s(u)|\overset{P_s}{\to}0\, .$$
    
        For testing the sharp null $P_s[\cap_{i=1}^{N_1}\{\alpha_i = c\}] = 1$ at significance level $\gamma$, the researcher could consider the test function:
		 \begin{equation}
		 	\label{eq_struct}
		 	\phi_c = \mathbf{1}\{\widehat{\boldsymbol{\alpha}} - c \leq \hat{Q}_s(\gamma/2)\}+ \mathbf{1}\{\widehat{\boldsymbol{\alpha}} - c \geq \hat{Q}_s(1-\gamma/2)\}\, .
		 \end{equation} 
         
		 This decision rule produces an asymptotic size $\gamma$-test of the sharp null if the distribution of  $\frac{1}{N_1}\sum_{i=1}^{N_1}\epsilon_i$ is continuous, by which we mean that:

         $$\lim_{s \to \infty} P_s[\phi_c] =  \gamma\, ,$$
         for a sequence of laws $(P_s)_s$ where the null holds (see the proof of the lemma below).
		 
		 The next result shows that, for tests with the structure \eqref{eq_struct}, test inversion produces an asymptotically valid prediction set under the same set of assumptions that justify the test.
		 
		 \begin{lemma}
         \label{lemma_equiv_few}
		 	Suppose that Assumption \ref{ass_cons_proxy} holds. Let $u \mapsto \hat{Q}_s(u)$ be a pointwise consistent estimator of the quantile function $Q_s$ of the distribution of the $\frac{1}{N_1}\sum_{i=1}^{N_1}\epsilon_i$, and consider the decision rule given by \eqref{eq_struct}. Define the set $\mathcal{B} = \{c \in \mathbb{R}: \phi_c = 0\}$. We then have that:
		 	\begin{equation}
		 		\label{eq_prediction_few_treated}
		 		\mathcal{B} = \left[\widehat{\boldsymbol{\alpha}}-\hat{Q}_s(1-\gamma/2), \widehat{\boldsymbol{\alpha}}-\hat{Q}_s(\gamma/2)\right] 
		 	\end{equation}
	 is an asymptotically valid $(1-\gamma)$-prediction set  if the distribution of $\frac{1}{N_1}\sum_{i=1}^{N_1}\epsilon_i$ is continuous, meaning that:

    $$\lim_{s \to \infty} P_s\left[\frac{1}{N_1}\sum_{i=1}^{N_1}\alpha_i \in \mathcal{B}\right] = 1-\gamma\, .$$
\begin{proof}
    See Appendix \ref{proof_lemma_equiv_few}.
\end{proof} 
    
		 \end{lemma}
		 
		 The previous result, when combined with Lemma \ref{lemma_imply},  shows that, for the usual approaches undertaken in the few treated literature, tests of sharp nulls and prediction sets are intricately connected, in the sense that prediction sets always produce (asymptotic) tests of sharp nulls, and tests of sharp nulls with the structure \eqref{eq_struct} can always be inverted to construct an asymptotically valid prediction set. 
		 
		 Finally, we consider the construction of asymptotically valid confidence intervals for the realised effect, by which we mean set-valued functions of the data $\mathcal{C}$ with the property that:

\begin{equation}
\label{eq_confidence_valid}
    \begin{aligned}
               P_s\left[\frac{1}{N_1}\sum_{i=1}^{N_1}\alpha_i \in \mathcal{C}\Big| \alpha_1,\ldots, \alpha_{N_1}\right] \overset{P_s}{\to} 1- \gamma\, .
    \end{aligned}
\end{equation}

A natural question to be asked is whether, in our triangular array setup where the distribution of the data varies with $s$, Assumption \ref{ass_cons_proxy} is equivalent to consistency, conditional on $(\alpha_1,\ldots \alpha_{N_1})$, of the proxies ``in probability'' -- the notion of consistency that one requires to construct sets satisfying \eqref{eq_confidence_valid}.
        The next lemma shows that this is in fact true. 
        
			\begin{lemma}
            \label{lemma_relate}
					$|\hat{M}_{i} -{M}_{i} | \overset{P_s}{\to} 0$ if, and only if, $P_s[|\hat{M}_{i} -{M}_{i} | > \epsilon|\alpha_{1}, \alpha_{2}, \ldots, \alpha_{N_1} ]  \overset{P_s}{\to} 0$ for every $\epsilon >0$.
					\begin{proof}
					    See Appendix \ref{proof_lemma_array}.
					\end{proof}
			\end{lemma}
            
            At its essence, the previous result reveals that the way we treat treatment effect heterogeneity -- either as stochastic or as ``realized and conditioned-on'' -- is ``irrelevant'' for assessing the correct notion of consistency of the proxies required by different inference methods. Indeed, our result shows that verifying if the proxies are consistent -- which is required for asymptotic validity of prediction intervals -- or conditionally consistent in probability -- which is required for validity of methods for inference on the realized effect -- always produces the same conclusion, given that both notions are equivalent.
         
         Given the result in Lemma \ref{lemma_relate}, it follows that, in constructing confidence sets for the realized effect, one could adopt a similar formulation to \eqref{eq_prediction_few_treated}, but replacing the estimator of the quantiles of the distribution of $\frac{1}{N_1}\sum_{i=1}^{N_1}\epsilon_i$ with a consistent estimator of the quantiles of the conditional on $(\alpha_i, \ldots, \alpha_{N_1})$ distribution of $\frac{1}{N_1}\sum_{i=1}^{N_1}\epsilon_i$. However, as argued in earlier sections, one difficulty of this approach is that consistency of the estimator of conditional quantiles requires assumptions on the relation between $\alpha_i$ and $\epsilon_i$ that may be hard to justify in practice. 
            
   In light of this difficulty, our final result provides a novel connection between confidence sets for the realized treatment effects and prediction sets. Specifically, we show that methods for conducting inference on treatment effects that rely on a consistent estimator of the distribution $e \mapsto P_s\left[\frac{1}{N_1}\sum_{i=1}^{N_1}\epsilon_i \leq e \right]$ can be either seen as valid prediction intervals, or as a valid method for inference on the realized effect, under the additional assumption that treatment effects are independent of the $\epsilon_i$. 
   
   \begin{proposition}
   \label{prop_results}
       Suppose that Assumption \ref{ass_cons_proxy} holds. Let $\hat \delta_s$ be a consistent estimator of a parameter $\delta_s$ belonging to a normed space $(\Delta,\lVert \cdot \rVert)$, i.e. $\lVert \hat \delta_s - \delta_s \rVert \overset{P_s}{\to} 0$. Define $\boldsymbol{X} = (X_i)_{i=1}^{N_1}$, and fix $\gamma \in (0,1/2)$. If the conditional distribution of $\frac{1}{N_1}\sum_{i=1}^{N_1}\epsilon_i$ is identified as, for every $e \in \mathbb{R}$:

       $$P_s\left[\frac{1}{N_1}\sum_{i=1}^{N_1}\epsilon_i \leq e \Big|\boldsymbol{X}\right] = \Psi_s(e|\boldsymbol{X}; \delta_s)\,,$$
with $\Psi_s(\cdot|\boldsymbol{X}; \delta_s)$ being a continuous distribution function with quantile function $q_s(\cdot|\boldsymbol{X})$ such that, for some $\nu, \kappa > 0$ and sequence $Z_s = O_{P_s}(1)$ of random variables:

\begin{equation}
    \begin{aligned}
           \label{eq_lip_cond} 
           \lVert \delta - \delta_s \rVert < \kappa \implies \\ \implies  \sup_{e \in \mathcal{E}} | \Psi_s(e|\boldsymbol{X}; \delta) - \Psi_s(e|\boldsymbol{X}; \delta_s)|\leq Z_s \lVert  \delta - \delta_s\rVert\, , 
    \end{aligned}
\end{equation}
where 
$$\mathcal{E} := (q_s(\gamma/2-\nu|\boldsymbol{X}), q_s(\gamma/2+\nu|\boldsymbol{X})+\nu) \cup (q_s(1-\gamma/2-\nu|\boldsymbol{X}), q_s(1-\gamma/2+\nu|\boldsymbol{X})+\nu) \, ,$$
we then have that:
\begin{equation*}
    P_s\left[\frac{1}{N_1}\sum_{i=1}\alpha_i \in [ \boldsymbol{\hat \alpha}- q(1-\gamma/2|\boldsymbol{X}; \hat\delta_s), \boldsymbol{\hat \alpha}- q(\gamma/2|\boldsymbol{X}; \hat\delta_s)\Big| \boldsymbol{X}\right] \overset{P_s}{\to} (1-\gamma)\, 
\end{equation*}

Alternatively, if the posited model as

   $$P_s\left[\frac{1}{N_1}\sum_{i=1}^{N_1}\epsilon_i \leq e \Big|\boldsymbol{X}, \alpha_{1}, \alpha_{2}, \ldots, \alpha_{N_1}\right] = \Psi_s(e|\boldsymbol{X}; \delta_s)\,,$$
   with $\Psi_s(\cdot|\boldsymbol{X}; \delta_s)$ being continuous with quantile function $q_s(\cdot|\boldsymbol{X})$ and satisfying \eqref{eq_lip_cond}, then:

   $$P_s\left[\frac{1}{N_1}\sum_{i=1}\alpha_i \in [ \boldsymbol{\hat \alpha}- q(1-\gamma/2|\boldsymbol{X}; \hat\delta_s), \boldsymbol{\hat \alpha}- q(\gamma/2| \boldsymbol{X}; \hat\delta_s)\Big| \boldsymbol{X}, \alpha_{1}, \alpha_{2}, \ldots, \alpha_{N_1}, \right] \overset{P_s}{\to} (1-\gamma)
  $$
   \end{proposition}
   \begin{proof}
       See Appendix \ref{proof_prop_main}.
   \end{proof}
\begin{example}
To illustrate our main result, it is useful to consider the setting of \cite{ferman2019inference}. The authors consider a situation with $N_0$ controls, for which observed outcomes correspond to the potential outcome under nontreatment status; and $N_1$ treated units, for which observed outcomes correspond to the potential outcome under treatment status. In what follows, we index random variables concerning the treated subgroup by the superscript $1$, and random variables pertaining to the controls by the superscript $0$. In our notation, \cite{ferman2019inference} adopt the following model for the potential outcomes under nontreatment:

$$Y^d_i(0) = \mu + h(X_i^d) \epsilon_i^d\,\quad d \in \{0,1\}, \ldots i =1,\ldots N_d\, , $$
where $X_i^d$ is an observable characteristic (in their motivating application where outcomes $Y_i^d$ denote group-averages, these correspond to group sizes), $\mu$ is a fixed constant -- which corresponds to the proxies $M_i^1$ for untreated potential outcomes of the treated group in our setting --, and the $\{\epsilon_i^0\}_{i=1}^{N_0} \cup \{\epsilon_i^1\}_{i=1}^{N_1}  $ are iid zero mean random variables with continuous distribution, independent from the $\{X_i^0\}_{i=1}^{N_0} \cup \{X_i^1\}_{i=1}^{N_1}$.\footnote{In the setting of \cite{ferman2019inference}, where the $Y^d_i$ correspond to the difference in average outcomes between a pre- and posttreatment periods, the assumption that the $\epsilon^d_i$ have zero-mean corresponds to a parallel trends assumption on the variation of untreated outcomes. The assumption that the $Y^d_i(0)$ follows a scale model may then be seen as a a strengthening of parallel trends to hold ``in distribution'', in the sense that, while one allows for different variances for $Y^d_i(0)$ in the treated and control groups according to the observed traits $X_i^d$, the distribution of untreated outcomes is, up to such scale changes, otherwise the same across groups. } 

Under a parametric model for the function $h(x)=h(x;\theta_0)$, and given a consistent estimator $\hat{\theta}$ of $\theta_0$ such that $\max_{d \in \{0,1\}} \max_{i=1\,,\ldots, N_d} |h(X^d_i;\theta_0)-h(X^d_i;\hat \theta)| = o_P(1)$, \cite{ferman2019inference} provide a consistent estimator for the conditional on $\boldsymbol{X}$ distribution of $\frac{1}{N_1}\sum_{i=1}^{N_1} \epsilon_{i}^1$. Their estimator consists in, first, estimating the normalized residuals $\hat{\xi}^0_i = \frac{Y^0_i - \hat{\mu}}{h(X_i;\hat \theta)}$,\footnote{In their paper, \cite{ferman2019inference} actually focus on estimators of the distribution that rely on ``null-imposed'' residuals that leverage the values of treated outcomes $Y_i^1$ subtracted of the value of the treatment effect under the null being tested. This increases the number of residuals from $N_0$ to $N_0+N_1$. Both estimators are asymptotically equivalent in the large $N_0$-limit that is the focus of their paper, and we therefore focus on the estimator that does not rely on treatment group residuals for conciseness.} where $\hat{\mu} = \frac{1}{N_0}\sum_{i=1}^{N_0} Y^0_i$. We then use these residuals to estimate the distribution of  $\frac{1}{N_1}\sum_{i=1}^{N_1} \epsilon_{i}^1$ conditional on  $\boldsymbol{X}$ as:

    $$\hat{\Psi}(c|\boldsymbol{X}) = \frac{1}{N_1^{N_0}}\sum_{j_1=1}^{N_0}  \sum_{j_2=1}^{N_0}\ldots \sum_{j_{N_1}=1}^{N_0} \mathbf{1}\left\{\frac{1}{N_1}\sum_{i=1}^{N_1} h(X_{i};\hat \theta)\hat{\xi}_{j_i}^0{h(X_{j_i};\hat \theta)} \leq c\right\}\, , c \in \mathbb{R} \, . $$

    Finally, for a given confidence level $1-\gamma$, we can construct an interval:

    $$\mathcal{I}_{1-\gamma}=\left[\frac{1}{N_1} \sum_{i=1}^{N_1} Y_i^1 - \hat{\mu} - \hat{q}(1-\gamma/2|\boldsymbol{X}), \frac{1}{N_1} \sum_{i=1}^{N_1} Y_i^1 - \hat{\mu} + \hat{q}(\gamma/2|\boldsymbol{X})\right]\,,$$
    where $\hat{q}$ denotes the quantile function obtained from inversion of $\hat{\Psi}$.

    While \cite{ferman2019inference} derived results in a setting with non-stochastic treatment effects, under the stated scale model for the untreated potential outcomes, and given the consistency of the parametric model for heteroskedasticity, Proposition \ref{prop_results} shows that $\mathcal{I}_{1-\gamma}$ can be simultaneously seen as an asymptotically valid prediction interval for $\frac{1}{N_1}\sum_{i=1}^{N_1} \alpha^1_i$; as well an asymptotically valid confidence interval for the realized treatment effect \emph{under the additional assumption} that the $ \{\alpha^1_i\}_i$ are independent from the $\{\epsilon^1_i\}_i$.
\end{example}

		\bibliographystyle{apalike}
		\bibliography{references}

        \appendix

        \section{Failure of prediction intervals in achieving conditional coverage}
\label{example_failure}
        We provide a simple example that illustrates how prediction intervals generally fail to generate conditional-on-treatment-effect coverage. Consider a simple setting where the outcome of a single treated unit is observed. Her untreated potential outcome is given by $Y(0)\sim N(0,1)$. Suppose that the treated potential outcome (which is observed) is given by $Y(1) = 2Y(0)$. In this case, a valid 95\% prediction interval for the treatment effect $\tau = Y(1) - Y(0)$ is given by:

        $$\mathcal{I} = [Y  - 1.96 , Y  + 1.96 ] \,,$$
        where $Y=Y(1)$ denotes the observed outcome. Validity follows immediately from observing that:

        $$\tau \in \mathcal{I}\iff Y(0) \in [-1.96 , + 1.96 ]\, , $$
        which implies that $\mathbb{P}[\tau \in \mathcal{I}] = 0.95$. 
        
        Notice, however, that the prediction interval $\mathcal{I}$ generally fails to cover the treatment effect conditionally on $\tau$. Indeed, in our simple example, $\tau = Y(0)$, which implies that $\mathbb{P}[\tau \in \mathcal{I}|\tau]$ is either zero or one. However, coverage is precisely one with probability 95\% (over possible realisations of $\tau$), and $0$ in the remainder 5\% of the possible cases, which generates the exact unconditional coverage of 95\% .

        \section{Counterexample of test statistic whose inversion does not produce a valid prediction set}

        Consider a simple cross-sectional setting with $N_1=2$ treated units, for which observed outcomes correspond to the potential outcome under treatment status, and a single control unit, for which the observed outcome corresponds to the potential outcome under nontreatment status. In what follows, we index random variables concerning the treated subgroup by the superscript $1$, and random variables pertaining to the control group by the superscript $0$.  
        
        We assume that $Y_1^0=Y_1^0(0)\sim N(\mu,1)$ and that $Y_i^1(0) \sim N(\mu,1)$, $i=1, 2$. Fix $c \in \mathbb{R}$. Consider the decision rule:

        $$\phi_c = \mathbf{1}\left\{\max_{i=1,2} |Y_i^1 - Y_1^0-c|>\iota_{1-\gamma}\right\},$$
        where $\iota_{1-\gamma}$ is the $(1-\gamma)$ quantile of $\max_{i=1,2}|Z_i|$, where $\begin{bmatrix}
            Z_1 \\
            Z_2
        \end{bmatrix} \sim \mathcal{N}\left(\begin{bmatrix}
            0 \\ 
            0
        \end{bmatrix}, \begin{bmatrix}
            2 & 1\\
            1 & 2
        \end{bmatrix}\right)$. This decision rule produces a valid test of the sharp null $P[\{\alpha_1^1=c\}\cap \{\alpha_2^1=c\}]=1$.

        Inverting the family of decision rules $\{\phi_c\}_c$ results in the region:

        $$\mathcal{C} = \begin{cases}
            \left[\max_{i=1,2} (Y_i^1 -Y^1_0) - \iota_{1-\gamma}, \min_{i=1,2} (Y_i^1 -Y^1_0) + \iota_{1-\gamma} \right] & \text{if }  |Y_1^1 - Y_2^1| \leq 2\iota_i \\
            \emptyset & \text{if }  |Y_1^1 - Y_2^1| > 2 \iota_i 
        \end{cases}$$

        Now, if $\alpha_i^1 = \delta Y_i^1(0)$, $Y_1^1 - Y_2^1 \sim N(0, 2(1+\delta)^2)$. Choosing $\delta$ sufficiently large, we can ensure that the probability of $\mathcal{C}$ being empty is larger than $\gamma$, thus showing that $\mathcal{C}$ is not a valid prediction set for the in-sample average treatment effect.
        \label{app_counter_equivalence}
\section{Proofs of results in Section \ref{eq_specialized}}

\subsection{Proof of Lemma \ref{lemma_equiv_few}}
\label{proof_lemma_equiv_few}
   \begin{proof}

       That the set $\mathcal{B}$ takes the form given by \eqref{eq_prediction_few_treated} is straightforward. 
  By the canonical representation \citep[page 26]{Vaart2023}, we may assume that the vector sequence of random variables, implicitly indexed by $s$, $\{\hat{M}_i,M_i,\alpha_i, \epsilon_i\}_{i=1}^{N_1}$, and $\hat{Q}_s(\gamma/2)$  and $\hat{Q}_s(1-\gamma/2)$ are all defined in a common probability space. Passing through subsequences if needed, we may also assume that the convergence in probability also holds almost surely \citep[Theorem 20.5]{Billingsley1995}. But then, since the distribution of $\frac{1}{N_1}\sum_{i=1}^{N_1}\epsilon_i$ is continuous, we have that, almost surely:

\begin{equation*}
    \begin{aligned}
         \liminf_{s\to \infty} \mathbf{1}\left\{Q_s(\gamma/2)\leq\frac{1}{N_1}\sum_{i=1}^{N_1}\epsilon_i \leq Q_s(1-\gamma/2)\right\} \leq 
        \\   \liminf_{s\to \infty}\mathbf{1}\left\{\frac{1}{N_1}\sum_{i=1}^{N_1}\alpha_i\in \mathcal{B}\right\}  \leq \limsup_{s\to \infty}\mathbf{1}\left\{\frac{1}{N_1}\sum_{i=1}^{N_1}\alpha_i\in \mathcal{B}\right\} \leq \\ \limsup_{s\to \infty} \mathbf{1}\left\{Q_s(\gamma/2)\leq\frac{1}{N_1}\sum_{i=1}^{N_1}\epsilon_i \leq Q_s(1-\gamma/2)\right\}
    \end{aligned}
\end{equation*}

The desired conclusion then follows from Fatou lemma  \cite[e.g.][Theorem 16.3]{Billingsley1995}.
        
    \end{proof}
        \subsection{Proof of Lemma \ref{lemma_relate} }
        \label{proof_lemma_array}
        \begin{proof}
		 Suppose that $|\hat{M}_{i} -{M}_{i} | \overset{P_s}{\to} 0$. Fix $\epsilon >0$ and $\xi > 0$. Markov inequality entails:
			
			$$P_s[P_s[|\hat{M}_{i} -{M}_{i} | > \epsilon|\alpha_{1}, \alpha_{2}, \ldots, \alpha_{N_1} ]  > \xi] \leq \frac{P_s[|\hat{M}_{i} -{M}_{i} | > \epsilon] }{\xi} \to 0 \, .$$
			
			For the other direction, suppose that, for each $\epsilon >0$, $P_s[|\hat{M}_{i} -{M}_{i} | > \epsilon|\alpha_{1}, \alpha_{2}, \ldots, \alpha_{N_1} ]  \overset{P_s}{\to} 0$. Let  $\mathbb{Q}_s$ denote the pushforward probability measure of $(|\hat{M}_{i} -{M}_{i} | ,\alpha_{1}, \alpha_{2}, \ldots, \alpha_{N_1})$ for index $s$. By the canonical representation \citep[page 26]{Vaart2023}, we may define  $\{(|\hat{M}_{i,s} -{M}_{i,s} | ,\alpha_{1,s}, \alpha_{2,s}, \ldots, \alpha_{N_1,s})\}_{s=1}^\infty$ in a common probability space with law $\mathbb{Q}$ such that the pushforward law of $(|\hat{M}_{i,s} -{M}_{i,s} | ,\alpha_{1,s}, \alpha_{2,s}, \ldots, \alpha_{N_1,s})$   coincides with $\mathbb{Q}_s$. It then follows by the bounded convergence theorem \citep[e.g.][page 64]{Durrett_2019} that $|\hat{M}_{i} -{M}_{i} | \overset{P_s}{\to} 0$, since, for every $\epsilon > 0$:
			
			$$P_s[|\hat{M}_{i} -{M}_{i} | > \epsilon] = \mathbb{Q}[\mathbb{Q}[|\hat{M}_{i,s} -{M}_{i,s} | > \epsilon|\alpha_{1,s}, \alpha_{2,s}, \ldots, \alpha_{N_1,s} ] ] \, ,$$
			with $\mathbb{Q}[|\hat{M}_{i,s} -{M}_{i,s} | > \epsilon|\alpha_{1,s}, \alpha_{2,s}, \ldots, \alpha_{N_1,s} ] $ being a bounded function that converges in probability to zero .
					\end{proof}

\subsection{Proof of Proposition \ref{prop_results}}
\begin{proof}
\label{proof_prop_main}

We begin by proving the first part of the proposition. By the canonical representation \citep[page 26]{Vaart2023} and the almost-sure representation theorem \citep[Theorem 1.10.3]{Vaart2023}, we may assume that the sequence of random variables is defined in the same probability space, and that, almost surely:

\begin{equation}\lim_{s\to \infty} \max_{i=1,\ldots, N_1} |\hat{M}_{i} -{M}_{i}| = 0\, ,\end{equation}
\begin{equation}
\label{uniform_as}
\lim_{s \to \infty}\max_{e \in \mathcal{E}} | \Psi_s(e|\boldsymbol{X}; \hat{\delta}_s) - \Psi_s(e|\boldsymbol{X}; \delta_s)|  = 0\, .\end{equation}

Fix $l < \nu$, and take $u \in \{\gamma/2,1-\gamma/2\}$. Item (a) of Lemma 21.1 of \cite{Vaart_1998} yields the following implications:

\begin{enumerate}
    \item[i]If $\Psi_s(q_s(u + l|\boldsymbol{X})|\boldsymbol{X}; \hat{\delta}_s) > u$, then $q_s(u + l|\boldsymbol{X}) \geq {q}_s(u|\boldsymbol{X};\hat{\delta}_s)$.
\item[ii] If $\Psi_s(q_s(u - l|\boldsymbol{X})|\boldsymbol{X}; \hat{\delta}_s) < u$, then  $q_s(u -l|\boldsymbol{X}) \leq {q}_s(u|\boldsymbol{X};\hat{\delta}_s)$.
\end{enumerate}

Combining the above with \eqref{uniform_as}, and using that $\Psi_s(q_s(u)|\boldsymbol{X})|\boldsymbol{X}; {\delta}_s) =u$ by continuity of the true distribution, we arrive at:
$$\limsup_{s \to \infty} {q}_s(u -l|\boldsymbol{X}) \leq \limsup_{s\to \infty} {q}_s(u|\boldsymbol{X};\hat{\delta}_s)  \leq \limsup_{s \to \infty} {q}_s(u +l|\boldsymbol{X})  $$

$$ \liminf_{s\to \infty} {q}_s(u-l|\boldsymbol{X})  
 \leq \liminf_{s\to \infty} {q}_s(u|\boldsymbol{X};\hat{\delta}_s)  \leq \liminf_{s \to \infty} {q}_s(u+l|\boldsymbol{X})  $$

Using the first of these inequalities, we obtain, by the conditional Fatou inequality \citep[page 88]{Williams_1991}

\begin{equation*}
\tiny
    \begin{aligned} 1-\gamma - 2l \leq \\
   \underbrace{ \liminf_{s \to \infty} E_s\left[\mathbf{1}\left\{\frac{1}{N_1}\sum_{i=1}\alpha_i \in [ \boldsymbol{\hat \alpha}- q(1-\gamma/2|\boldsymbol{X};\hat \delta_s), \boldsymbol{\hat \alpha}- q(\gamma/2|\boldsymbol{X};\hat \delta_s)\right\}-\mathbf{1}\left\{\frac{1}{N_1}\sum_{i=1}\alpha_i \in [ \boldsymbol{\hat \alpha}- q(1-\gamma/2 -l|\boldsymbol{X}), \boldsymbol{\hat \alpha}- q(\gamma/2+l|\boldsymbol{X})\right\}\Big| \boldsymbol{X}\right]}_{\overset{\text{Fatou}}{\geq}0}
    \\ +
\liminf_{s \to \infty} \underbrace{E_s\left[\mathbf{1}\left\{\frac{1}{N_1}\sum_{i=1}\alpha_i \in [ \boldsymbol{\hat \alpha}- q(1-\gamma/2 -l|\boldsymbol{X}), \boldsymbol{\hat \alpha}- q(\gamma/2+l|\boldsymbol{X})\right\}\Big| \boldsymbol{X}\right]}_{\overset{\text{continuity of true distr.}}{=}1-\gamma - 2l} \leq
\\ \liminf_{s \to \infty}  P_s\left[\frac{1}{N_1}\sum_{i=1}\alpha_i \in [ \boldsymbol{\hat \alpha}- q(1-\gamma/2|\boldsymbol{X};\hat\delta_s), \boldsymbol{\hat \alpha}- q(\gamma/2|\boldsymbol{X};\hat\delta_s)\Big| \boldsymbol{X}\right]
    \end{aligned}
\end{equation*}

    By a similar argument, we obtain, by application of the conditional (reverse) Fatou inequality that:
    $$ \limsup_{s \to \infty}  P_s\left[\frac{1}{N_1}\sum_{i=1}\alpha_i \in [ \boldsymbol{\hat \alpha}- q(1-\gamma/2|\boldsymbol{X};\hat\delta_s), \boldsymbol{\hat \alpha}- q(\gamma/2|\boldsymbol{X};\hat\delta_s)\Big| \boldsymbol{X}\right] \leq 1-\gamma + 2l\, ,$$
    and since the choice of $l$ was arbitrarily, we obtain that, almost surely: 

    $$\lim_{s \to \infty}  P_s\left[\frac{1}{N_1}\sum_{i=1}\alpha_i \in [ \boldsymbol{\hat \alpha}- q(1-\gamma/2|\boldsymbol{X};\hat\delta_s), \boldsymbol{\hat \alpha}- q(\gamma/2|\boldsymbol{X};\hat\delta_s)\Big| \boldsymbol{X}\right] = 1-\gamma\, ,$$
which implies the convergence in probability stated in the proposition.

    The proof of the second part is identical -- one replaces conditional expectations and probabilities wrt $\boldsymbol{X}$ with conditioning on $\boldsymbol{X}$ and the treatment effect -- and is therefore omitted.
\end{proof}

\pagebreak

\end{document}